\documentclass[conference,letterpaper]{IEEEtran}

\usepackage[utf8]{inputenc} 
\usepackage[T1]{fontenc}
\usepackage{url}
\usepackage{ifthen}
\usepackage{cite}
\usepackage[cmex10]{amsmath} 
                             
\usepackage{amsfonts}
\usepackage{amssymb}%
\usepackage{hyperref}%
\usepackage{graphicx}

\newtheorem{theorem}{Theorem}

\newtheorem{lemma}[theorem]{Lemma}

\newenvironment{proof}[1][Proof]{\noindent\textbf{#1.} }{\ \rule{0.5em}{0.5em}}

\allowdisplaybreaks

\usepackage[]{color}

\begin{document}
\title{Entropy Bound for the Classical Capacity of a Quantum Channel Assisted by Classical Feedback} 



 \author{%
   \IEEEauthorblockN{Dawei Ding\IEEEauthorrefmark{1},
                     Yihui Quek\IEEEauthorrefmark{2},
                     Peter W.~Shor\IEEEauthorrefmark{3},
                     and Mark M.~Wilde\IEEEauthorrefmark{4}}
   \IEEEauthorblockA{\IEEEauthorrefmark{1}%
                     Stanford Institute for Theoretical Physics,
				     Stanford University, Stanford, California 94305, USA,
                     dding@stanford.edu}
   \IEEEauthorblockA{\IEEEauthorrefmark{2}%
                     Information Systems Laboratory, Stanford University, Stanford, California 94305, USA,
                     yquek@stanford.edu}
   \IEEEauthorblockA{\IEEEauthorrefmark{3}%
                     Center for Theoretical Physics, Massachusetts Institute of Technology, Cambridge, Massachusetts 02139, USA,\\
					Department of Mathematics, Massachusetts Institute of Technology, Cambridge, Massachusetts 02139, USA,
                     shor@math.mit.edu}
   \IEEEauthorblockA{\IEEEauthorrefmark{4}%
                     Hearne Institute for Theoretical Physics, Department of Physics and Astronomy, Center for Computation and Technology,\\Louisiana State University, Baton Rouge, Louisiana 70803, USA,
                     mwilde@lsu.edu}
 }

\maketitle

\begin{abstract}
  We prove that the classical capacity of an arbitrary quantum channel assisted by
a free classical feedback channel is bounded from above by the maximum average output entropy of the quantum channel.
As a consequence of this bound, we conclude that a classical feedback channel does not improve the classical capacity of a quantum erasure channel, and by taking into account energy constraints, we conclude the same for a pure-loss bosonic channel.
The method for establishing the aforementioned entropy bound involves identifying an information measure having two key properties: 1) it does not increase under a one-way
local operations and classical communication channel from the receiver to the sender and 2) a quantum channel from sender to receiver
cannot increase the information measure by more than the maximum output entropy of the channel. This information measure can be understood as the sum of two terms, with one corresponding to classical correlation and the other to entanglement. 
\end{abstract}


\section{Introduction}

A famous result of Shannon is that a free feedback channel does not increase the capacity of a
classical channel for communication \cite{S56}. That is, the feedback-assisted capacity is equal to
the channel's mutual information. Shannon's result indicates that the mutual information
formula for capacity is particularly robust, in the sense that, \textit{a priori}, one might
consider a feedback channel to be a strong resource for assisting communication.

With the rise of quantum information theory, 
several researchers have found variations and generalizations of Shannon's
aforementioned result, in the context of communication over quantum channels. For example, Bowen proved that 
the capacity of a quantum channel for sending classical messages, when assisted by a free quantum feedback channel,
is equal to the channel's entanglement-assisted capacity \cite{B04}, which is in turn equal to the mutual information
of a quantum channel \cite{PhysRevLett.83.3081,ieee2002bennett,Hol01a}. This result indicates that the mutual information
of a quantum channel is robust, in a sense similar to that mentioned above. The result also indicates that the best strategy, in the limit of many channel uses, is to use the quantum feedback channel once in order to establish sufficient shared entanglement between the sender and receiver, and to subsequently employ an entanglement-assisted communication protocol \cite{PhysRevLett.83.3081,ieee2002bennett,Hol01a}.
Bowen's result was strengthened
to a strong converse statement in \cite{BDHSW09,CMW14}.
Bowen \textit{et al.} proved that the capacity of an entanglement-breaking channel for sending classical messages is not increased by a free classical feedback channel \cite{BN05}, and this result was strengthened to a strong converse statement in \cite{DingW18}. Ref.~\cite{BDSS06} discussed several inequalities relating the classical capacity assisted by classical feedback to other capacities. At the same time, it is known that in general there can be an arbitrarily large gap between the unassisted classical capacity and the  classical capacity assisted by classical feedback \cite{SS09}.

Our aim here is to go beyond \cite{BN05} to establish
an upper bound on the classical capacity of an \textit{arbitrary}, not just entanglement-breaking, quantum channel
assisted by a classical feedback channel. Due to the fact that a quantum feedback channel is a stronger resource than a classical feedback channel, an immediate consequence of Bowen's result \cite{B04} is that the entanglement-assisted capacity is an upper bound on the classical capacity assisted by classical feedback.
However, since a quantum channel
can, in general, establish quantum entanglement \cite{PhysRevA.55.1613,capacity2002shor,ieee2005dev} and entanglement can increase capacity \cite{PhysRevLett.83.3081,ieee2002bennett,Hol01a}, in such cases it may appear difficult to establish an upper bound on this capacity other than the entanglement-assisted capacity.
Our main result is that the latter is actually possible: we prove here that the maximum output entropy of a quantum channel is an upper bound on its classical
capacity assisted by classical feedback. As a generalization of this result, we find that the maximum average output entropy is an upper bound on the same capacity for a channel that is a probabilistic mixture of other channels.

The approach that we take for establishing the aforementioned bounds
is similar in spirit to approaches used to bound other assisted capacities or protocols \cite{TGW14,CM17,KW17,BHKW18}.
We identify an information measure that has two key properties: 1)  it does not increase under a free operation, which in this case is a one-way
local operations and classical communication (1W-LOCC)
channel from the receiver to the sender, and 2) a quantum channel from sender to receiver
cannot increase the information measure by more than the maximum output entropy of the channel. This information measure can be understood as the sum of two terms, with one corresponding to classical correlation and the other to entanglement.

We organize the rest of the paper as follows. Section~\ref{sec:feedback-protocol} provides a formal definition
of a protocol for classical communication over a quantum channel assisted by classical feedback.
Section~\ref{sec:purified-protocol} discusses explicitly how to purify such a protocol, which is an important conceptual step for our analysis. Section~\ref{sec:info-measure}  introduces our key information measure and several important supplementary lemmas regarding it. Section~\ref{sec:max-out-ent-bnd} then employs this information measure and the supplementary lemmas to establish the maximum output entropy bound for classical capacity assisted by classical feedback.
We apply this bound to the erasure channel and pure-loss bosonic channel in Section~\ref{sec:examples}. We conclude in Section~\ref{sec:conclusion}.

\section{Protocol for classical communication over a quantum channel assisted
by classical feedback}

\label{sec:feedback-protocol}

To begin with, let $n,M\in\mathbb{N}$, let
$\varepsilon\in\left[  0,1\right]  $, and let $E\geq0$. Let $\mathcal{N}%
_{A\rightarrow B}$ be a quantum channel, and let $H$ be a Hamiltonian acting
on the input system $A$ of $\mathcal{N}_{A\rightarrow B}$. An
$(n,M,H,E,\varepsilon)$ protocol for classical communication over a quantum
channel $\mathcal{N}_{A\rightarrow B}$ consists of $n$ uses of the quantum
channel $\mathcal{N}_{A\rightarrow B}$, along with the assistance of a
classical feedback channel from the receiver Bob to the sender Alice, in order for Alice to send one
of $M$ messages to Bob such that the error probability is no larger than
$\varepsilon$. Furthermore, the average state at the input of each channel use
should have energy no larger than $E$, when taken with respect to the
Hamiltonian $H$.

In more detail, the protocol consists of an initial classical--quantum state
$\sigma_{F_{0}B_{1}^{\prime}}$, with $F_{0}$ classical and $B_{1}^{\prime}$
quantum, of the form%
\begin{equation}
\sigma_{F_{0}B_{1}^{\prime}}=\sum_{f_{0}}p(f_{0})|f_{0}\rangle\langle
f_{0}|_{F_{0}}\otimes\sigma_{B_{1}^{\prime}}^{f_{0}}.
\end{equation}
It also involves $n$ encoding channels, with each one denoted by
$\mathcal{E}_{A_{i-1}^{\prime}F_{i-1}\rightarrow A_{i}^{\prime}A_{i}}^{i}$ for
$i\in\left\{  1,\ldots,n\right\}  $, as well as $n$ decoding channels, with
each of them denoted by $\mathcal{D}_{B_{i}B_{i}^{\prime}\rightarrow
F_{i}B_{i+1}^{\prime}}^{i}$ for $i\in\left\{  1,\ldots,n-1\right\}  $. Note
that all $F$ systems are classical because the feedback channel is constrained
to be a classical channel. So this means that each decoding channel is a
quantum instrument. The final decoding is denoted by $\mathcal{D}_{B_{n}%
B_{n}^{\prime}\rightarrow\hat{W}}^{n}$.

We now detail the form of such a protocol. It begins with Alice preparing the following
classical--quantum state:%
\begin{equation}
\rho_{WA_{0}^{\prime}}\equiv\frac{1}{M}\sum_{m=1}^{M}|m\rangle\langle
m|_{W}\otimes\rho_{A_{0}^{\prime}}^{m},
\end{equation}
for some set $\{\rho_{A_{0}^{\prime}}^{m}\}_{m}$ of quantum states. The global
initial state is then $\rho_{WA_{0}^{\prime}}\otimes\sigma_{F_{0}B_{1}%
^{\prime}}$. Alice then performs the encoding channel $\mathcal{E}%
_{A_{0}^{\prime}F_{0}\rightarrow A_{1}^{\prime}A_{1}}^{1}$ and the state
becomes as follows:%
\begin{equation}
\omega_{WA_{1}^{\prime}A_{1}B_{1}^{\prime}}^{(1)}\equiv\mathcal{E}%
_{A_{0}^{\prime}F_{0}\rightarrow A_{1}^{\prime}A_{1}}^{1}(\rho_{WA_{0}%
^{\prime}}\otimes\sigma_{F_{0}B_{1}^{\prime}}).
\end{equation}
Alice transmits the $A_{1}$ system through the first use of the channel
$\mathcal{N}_{A_{1}\rightarrow B_{1}}$, resulting in the following state:%
\begin{equation}
\rho_{WA_{1}^{\prime}B_{1}B_{1}^{\prime}}^{(1)}\equiv \mathcal{N}_{A_{1}\rightarrow
B_{1}}(\omega^{(1)}_{WA_{1}^{\prime}A_{1}B_{1}^{\prime}}).
\end{equation}
Bob processes his systems $B_{1}B_{1}^{\prime}$ with the decoding channel
$\mathcal{D}_{B_{1}B_{1}^{\prime}\rightarrow F_{1}B_{2}^{\prime}}^{1}$ and
Alice acts with the encoding channel $\mathcal{E}_{A_{1}^{\prime}%
F_{1}\rightarrow A_{2}^{\prime}A_{2}}^{2}$, resulting in the state%
\begin{equation}
\omega_{WA_{2}^{\prime}A_{2}B_{2}^{\prime}}^{(2)}\equiv(\mathcal{E}%
_{A_{1}^{\prime}F_{1}\rightarrow A_{2}^{\prime}A_{2}}^{2}\circ\mathcal{D}%
_{B_{1}B_{1}^{\prime}\rightarrow F_{1}B_{2}^{\prime}}^{1})(\rho_{WA_{1}%
^{\prime}B_{1}B_{1}^{\prime}}^{(1)}).
\end{equation}
This process iterates $n-2$ more times, resulting in the following states:%
\begin{equation}
\rho_{WA_{i}^{\prime}B_{i}B_{i}^{\prime}}^{(i)}   \equiv\mathcal{N}%
_{A_{i}\rightarrow B_{i}}(\omega_{WA_{i}^{\prime}A_{i}B_{i}^{\prime}}%
^{(i)}),
\end{equation}
\vspace{-.25in}
\begin{multline}
\omega_{WA_{i+1}^{\prime}A_{i+1}B_{i+1}^{\prime}}^{(i+1)}  \equiv 
\\
(\mathcal{E}_{A_{i}^{\prime}F_{i}\rightarrow A_{i+1}^{\prime}A_{i+1}}%
^{i+1}\circ\mathcal{D}_{B_{i}B_{i}^{\prime}\rightarrow F_{i}B_{i+1}^{\prime}%
}^{i})(\rho_{WA_{i}^{\prime}B_{i}B_{i}^{\prime}}^{(i)}),
\end{multline}
for $i\in\left\{  2,\ldots,n-1\right\}  $. The final decoding
(measurement)\ channel $\mathcal{D}_{B_{n}B_{n}^{\prime}\rightarrow\hat{W}%
}^{n}$\ results in the following state:%
\begin{equation}
\rho_{W\hat{W}}\equiv(\operatorname{Tr}_{A_{n}^{\prime}}\circ\mathcal{D}%
_{B_{n}B_{n}^{\prime}\rightarrow\hat{W}}^{n})(\rho_{WA_{n}^{\prime}B_{n}%
B_{n}^{\prime}}^{(n)}).
\end{equation}
Figure~\ref{fig:feedback-prot} depicts the above protocol for $n=3$.%
\begin{figure}
[ptb]
\begin{center}
\includegraphics[
width=3.7in
]%
{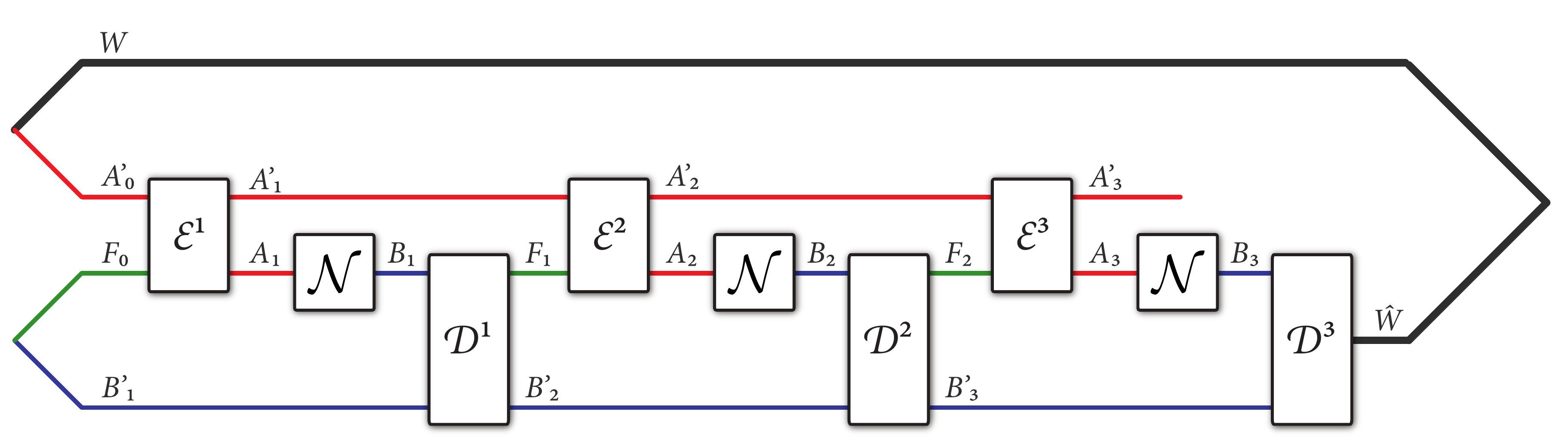}%
\caption{A protocol for classical communication over three uses of a quantum channel $\mathcal{N}_{A \to B}$, when
assisted by a classical feedback channel.}%
\label{fig:feedback-prot}%
\end{center}
\end{figure}

For an $(n,M,H,E,\varepsilon)$ protocol, the following is satisfied%
\begin{equation}
\tfrac{1}{2}\left\Vert \overline{\Phi}_{W\hat{W}}-\rho_{W\hat{W}}\right\Vert
_{1}\leq\varepsilon,\label{eq:good-protocol-condition}%
\end{equation}
where $\overline{\Phi}_{W\hat{W}}\equiv\frac{1}{M}\sum_{m=1}^{M}%
|m\rangle\langle m|_{W}\otimes|m\rangle\langle m|_{\hat{W}}$ is the maximally
classically correlated state. Note that
$
\tfrac{1}{2}\left\Vert \overline{\Phi}_{W\hat{W}}-\rho_{W\hat{W}}\right\Vert
_{1}=\Pr\{\hat{W}\neq W\}$,
where $W$ here denotes the uniform random variable corresponding to the
message choice and $\hat{W}$ denotes the random variable corresponding to the
classical value in the register $\hat{W}$ of the state $\rho_{W\hat{W}}$.
Furthermore, the following energy constraint applies as well:%
\begin{equation}
\operatorname{Tr}\{H\overline{\omega}_{A}\}  \leq
E,\label{eq:energy-constraint} \qquad
\overline{\omega}_{A}   \equiv\frac{1}{n}\sum_{i=1}^{n}\omega_{A_{i}}%
^{(i)},
\end{equation}
which limits the energy of the average input state.

\section{Purified protocol}

\label{sec:purified-protocol}

Our goal is to bound the rate of such a protocol. With this in mind, we can
devise a protocol that simulates the above one. It consists of purifying each
step of the above protocol and Bob keeping a copy of the classical feedback,
such that at each time step, conditioned on the value of the message in $W$
and the feedback in the existing systems labeled by $F$, the state is pure. To
be clear, we go through the steps of the purified protocol. In order to
simplify notation, we let $\hat{A}$ be a joint system throughout, referring to
both the original $A^{\prime}$ system as well as a purifying system, and we
take the same convention for $\hat{B}$. The initial state of Alice is as
follows:%
\begin{equation}
\rho_{W\hat{A}_{0}}\equiv\frac{1}{M}\sum_{m=1}^{M}|m\rangle\langle
m|_{W}\otimes\psi_{\hat{A}_{0}}^{m},
\end{equation}
where $\psi_{\hat{A}_{0}}^{m}$ is a purification of $\rho_{A_{0}^{\prime}}%
^{m}$, such that tracing over a subsystem of $\psi_{\hat{A}_{0}}^{m}$ gives
$\rho_{A_{0}^{\prime}}^{m}$. The initial state of Bob is as follows:%
\begin{equation}
\sigma_{F_{0}F_{0}^{\prime}\hat{B}_{1}}\equiv\sum_{f_{0}}p(f_{0})|f_{0}%
\rangle\langle f_{0}|_{F_{0}}\otimes|f_{0}\rangle\langle f_{0}|_{F_{0}%
^{\prime}}\otimes\varphi_{\hat{B}_{1}}^{f_{0}},
\end{equation}
where $\varphi_{\hat{B}_{1}}^{f_{0}}$ is a purification of $\sigma
_{B_{1}^{\prime}}^{f_{0}}$, such that tracing over a subsystem of
$\varphi_{\hat{B}_{1}}^{f_{0}}$ gives $\sigma_{B_{1}^{\prime}}^{f_{0}}$, and
he keeps an extra copy $F_{0}^{\prime}$ of the classical data. Let
$\mathcal{U}_{\hat{A}_{i-1}F_{i-1}\rightarrow\hat{A}_{i}A_{i}}^{i}$ denote an
isometric channel extending the encoding channel $\mathcal{E}_{A_{i-1}%
^{\prime}F_{i-1}\rightarrow A_{i}^{\prime}A_{i}}^{i}$, for $i\in\left\{
1,\ldots,n\right\}  $. Since the system $F_{i}$ is classical, for
$i\in\left\{  1,\ldots,n-1\right\}  $, the decoding channel $\mathcal{D}%
_{B_{i}B_{i}^{\prime}\rightarrow F_{i}B_{i+1}^{\prime}}^{i}$ can be written
explicitly as%
\begin{equation}
\mathcal{D}_{B_{i}B_{i}^{\prime}\rightarrow F_{i}B_{i+1}^{\prime}}^{i}%
=\sum_{f_{i}}\mathcal{D}_{B_{i}B_{i}^{\prime}\rightarrow B_{i+1}^{\prime}%
}^{i,f_{i}}\otimes|f_{i}\rangle\langle f_{i}|_{F_{i}},
\end{equation}
such that $\{\mathcal{D}_{B_{i}B_{i}^{\prime}\rightarrow B_{i+1}^{\prime}%
}^{i,f_{i}}\}_{f_{i}}$\ is a collection of completely positive maps such that
the sum map $\sum_{f_{i}}\mathcal{D}_{B_{i}B_{i}^{\prime}\rightarrow
B_{i+1}^{\prime}}^{i,f_{i}}$ is trace preserving. Let $V_{B_{i}\hat{B}%
_{i}\rightarrow\hat{B}_{i+1}}^{i,f_{i}}$ be a linear map such that tracing
over a subsystem of $V_{B_{i}\hat{B}_{i}\rightarrow\hat{B}_{i+1}}^{i,f_{i}%
}(\cdot)[V_{B_{i}\hat{B}_{i}\rightarrow\hat{B}_{i+1}}^{i,f_{i}}]^{\dag}$ gives
the original map $\mathcal{D}_{B_{i}B_{i}^{\prime}\rightarrow B_{i+1}^{\prime
}}^{i,f_{i}}$, and define the map%
\begin{equation}
\mathcal{V}_{B_{i}\hat{B}_{i}\rightarrow\hat{B}_{i+1}}^{i,f_{i}}(\tau
_{B_{i}\hat{B}_{i}})\equiv V_{B_{i}\hat{B}_{i}\rightarrow\hat{B}_{i+1}%
}^{i,f_{i}}\tau_{B_{i}\hat{B}_{i}}[V_{B_{i}\hat{B}_{i}\rightarrow\hat{B}%
_{i+1}}^{i,f_{i}}]^{\dag}.\notag
\end{equation}
Then we define the enlarged decoding channel $\mathcal{V}_{B_{i}\hat{B}%
_{i}\rightarrow F_{i}\hat{B}_{i+1}F_{i}^{\prime}}^{i}$ as%
\begin{equation}
\mathcal{V}_{B_{i}\hat{B}_{i}\rightarrow F_{i}\hat{B}_{i+1}F_{i}^{\prime}}%
^{i}\equiv\sum_{f_{i}}\mathcal{V}_{B_{i}\hat{B}_{i}\rightarrow\hat{B}_{i+1}%
}^{i,f_{i}}\otimes|f_{i}\rangle\langle f_{i}|_{F_{i}}\otimes|f_{i}%
\rangle\langle f_{i}|_{F_{i}^{\prime}}.\notag
\end{equation}
Note that this enlarged decoding channel keeps an extra copy of the classical
feedback value for Bob in the register $F_{i}^{\prime}$. The final decoding
channel in the original protocol is equivalent to a measurement channel, and
thus can be written as%
\begin{equation}
\mathcal{D}_{B_{n}B_{n}^{\prime}\rightarrow\hat{W}}^{n}(\tau_{B_{n}%
B_{n}^{\prime}})=\sum_{w}\operatorname{Tr}\{\Lambda_{B_{n}B_{n}^{\prime}}%
^{w}\tau_{B_{n}B_{n}^{\prime}}\}|w\rangle\langle w|_{\hat{W}},\notag
\end{equation}
where $\{\Lambda_{B_{n}B_{n}^{\prime}}^{w}\}_{w}$ is a POVM. We enlarge it as
follows in the simulation protocol:%
\begin{multline}
\mathcal{V}_{B_{n}\hat{B}_{n}\rightarrow\hat{B}_{n+1}\hat{W}}^{n}(\tau
_{B_{n}\hat{B}_{n}})= \\ \sum_{w}\sqrt{\Lambda_{B_{n}B_{n}^{\prime}}^{w}}%
\tau_{B_{n}\hat{B}_{n}}\sqrt{\Lambda_{B_{n}B_{n}^{\prime}}^{w}}\otimes
|w\rangle\langle w|_{\hat{W}},
\end{multline}
where the meaning of the notation is that the map $\sqrt{\Lambda_{B_{n}%
B_{n}^{\prime}}^{w}}(\cdot)\sqrt{\Lambda_{B_{n}B_{n}^{\prime}}^{w}}$ acts
nontrivially on the subsystems $B_{n}B_{n}^{\prime}$ in the original protocol
and trivially on all other $B$ subsystems, while mapping all $B$ systems to a
system $\hat{B}_{n+1}$ large enough to accommodate all of them. In the
simulation protocol, we also consider an isometric channel $\mathcal{U}%
_{A\rightarrow BE}^{\mathcal{N}}$ that simulates the original channel
$\mathcal{N}_{A\rightarrow B}$ as follows:
$
\mathcal{N}_{A\rightarrow B}=\operatorname{Tr}_{E}\circ\mathcal{U}%
_{A\rightarrow BE}^{\mathcal{N}}$.

Thus, the various states involved in the purified protocol are as follows. The
global initial state is $\rho_{W\hat{A}_{0}}\otimes\sigma_{F_{0}F_{0}^{\prime
}\hat{B}_{1}}$. Alice performs the enlarged encoding channel $\mathcal{U}%
_{\hat{A}_{0}F_{0}\rightarrow\hat{A}_{1}A_{1}}^{1}$ and the state becomes as
follows:%
\begin{equation}
\omega_{W\hat{A}_{1}A_{1}\hat{B}_{1}F_{0}^{\prime}}^{(1)}\equiv\mathcal{U}%
_{\hat{A}_{0}F_{0}\rightarrow\hat{A}_{1}A_{1}}^{1}(\rho_{W\hat{A}_{0}}%
\otimes\sigma_{F_{0}F_{0}^{\prime}\hat{B}_{1}}).
\end{equation}
Alice transmits the $A_{1}$ system through the first use of the extended
channel $\mathcal{U}_{A_{1}\rightarrow B_{1}E_{1}}^{\mathcal{N}}$, resulting
in the following state:%
\begin{equation}
\rho_{W\hat{A}_{1}B_{1}\hat{B}_{1}E_{1}F_{0}^{\prime}}^{(1)}=\mathcal{U}%
_{A_{1}\rightarrow B_{1}E_{1}}^{\mathcal{N}}(\omega^{(1)}_{W\hat{A}_{1}A_{1}\hat
{B}_{1}F_{0}^{\prime}}).
\end{equation}
Bob processes his systems $B_{1}\hat{B}_{1}$ with the enlarged decoding
channel $\mathcal{V}_{B_{1}\hat{B}_{1}\rightarrow F_{1}F_{1}^{\prime}\hat
{B}_{2}}^{1}$ and Alice acts with the enlarged encoding channel $\mathcal{U}%
_{\hat{A}_{1}F_{1}\rightarrow\hat{A}_{2}A_{2}}^{2}$, resulting in the state
$
\omega_{W\hat{A}_{2}A_{2}\hat{B}_{2}E_{1}F_{0}^{\prime}F_{1}^{\prime}}%
^{(2)}\equiv
(\mathcal{U}_{\hat{A}_{1}F_{1}\rightarrow\hat{A}_{2}A_{2}}%
^{2}\circ\mathcal{V}_{B_{1}\hat{B}_{1}\rightarrow F_{1}F_{1}^{\prime}\hat
{B}_{2}}^{1})(\rho_{W\hat{A}_{1}B_{1}\hat{B}_{1}E_{1}F_{0}^{\prime}}^{(1)})$.
This process iterates $n-2$ more times, resulting in the following states:%
\begin{equation}
\rho_{W\hat{A}_{i}B_{i}\hat{B}_{i}E_{1}^{i}[F_{0}^{i-1}]^{\prime}}^{(i)}  
\equiv\mathcal{U}_{A_{i}\rightarrow B_{i}E_{i}}^{\mathcal{N}}(\omega_{W\hat
{A}_{i}A_{i}\hat{B}_{i}E_{1}^{i-1}[F_{0}^{i-1}]^{\prime}}^{(i)}),\notag
\end{equation}
\vspace{-.30in}
\begin{multline}
\omega_{W\hat{A}_{i+1}A_{i+1}\hat{B}_{i+1}E_{1}^{i}[F_{0}^{i}]^{\prime}%
}^{(i+1)}   \equiv
\\
(\mathcal{U}_{\hat{A}_{i}F_{i}\rightarrow\hat{A}%
_{i+1}A_{i+1}}^{i+1}\circ\mathcal{V}_{B_{i}\hat{B}_{i}\rightarrow F_{i}%
F_{i}^{\prime}\hat{B}_{i+1}}^{i})(\rho_{W\hat{A}_{i}B_{i}\hat{B}_{i}E_{1}%
^{i}[F_{0}^{i-1}]^{\prime}}^{(i)}),\notag
\end{multline}
for $i\in\left\{  2,\ldots,n-1\right\}  $. The final enlarged decoding
 channel $\mathcal{V}_{B_{n}\hat{B}_{n}\rightarrow\hat{B}_{n+1}
\hat{W}}^{n}$\ results in the following state:
$
\rho_{W\hat{A}_{n}\hat{B}_{n+1}\hat{W}E_{1}^{n}[F_{0}^{n-1}]^{\prime}}%
\equiv
\mathcal{V}_{B_{n}\hat{B}_{n}\rightarrow\hat{B}_{n+1}\hat{W}}^{n}%
(\rho_{W\hat{A}_{n}B_{n}\hat{B}_{n}E_{1}^{n}[F_{0}^{n-1}]^{\prime}}^{(n)})$.
Note that we recover each state of the original protocol from
Section~\ref{sec:feedback-protocol}\ by performing particular partial traces.

\section{Information measure for analysis of protocol}

\label{sec:info-measure}

The key information measure that we use to analyze this protocol is as
follows:%
\begin{equation}
I(W;CF)_{\tau}+S(C|WF)_{\tau},\label{eq:1W-LOCC-monotone}%
\end{equation}
where $\tau_{WFC}$ is a classical--quantum state of the form
\begin{equation}
\tau_{WFC}=\sum_{w,f}p(w,f)|w\rangle\langle w|_{W}\otimes|f\rangle\langle
f|_{F}\otimes\tau_{C}^{w,f}.
\end{equation}
The first term in \eqref{eq:1W-LOCC-monotone} represents the classical
correlation between system $W$ and systems $CF$, while the second term
represents the average entanglement between the system $C$ of the state
$\tau_{C}^{w,f}$ and a purifying reference system.



We now establish some properties of the information measure in \eqref{eq:1W-LOCC-monotone}. Let us first recall the following lemma from \cite{BDSW96}:

\begin{lemma}
\label{lem:LOCC-monotone}Let $\phi_{AB}$ be a pure bipartite state, and let
$\{p(x),\varphi_{A^{\prime}B^{\prime}}^{x}\}$ be an ensemble of pure bipartite
states obtained from $\phi_{AB}$ by means of a 1W-LOCC channel of the
form
\begin{equation}
\sum_{x}\mathcal{U}_{A\rightarrow A^{\prime}}^{x}\otimes\mathcal{V}%
_{B\rightarrow B^{\prime}}^{x}\otimes|x\rangle\langle x|_{X}%
,\label{eq:1W-LOCC-ch}%
\end{equation}
where $\{\mathcal{V}_{B\rightarrow B^{\prime}}^{x}\}_{x}$ is a collection of
completely positive trace non-increasing maps with $\mathcal{V}%
_{B\rightarrow B^{\prime}}^{x}(\cdot)=V_{B\rightarrow B^{\prime}}^{x}%
(\cdot)[V_{B\rightarrow B^{\prime}}^{x}]^{\dag}$ and $\{\mathcal{U}%
_{A\rightarrow A^{\prime}}^{x}\}_{x}$ is a collection of isometric channels,
so that%
\begin{align}
\varphi_{A^{\prime}B^{\prime}}^{x}  & \equiv\frac{1}{p(x)}(\mathcal{U}%
_{A\rightarrow A^{\prime}}^{x}\otimes\mathcal{V}_{B\rightarrow B^{\prime}}%
^{x})(\phi_{AB}),\\
p(x)  & \equiv\operatorname{Tr}\{(\mathcal{U}_{A\rightarrow A^{\prime}}%
^{x}\otimes\mathcal{V}_{B\rightarrow B^{\prime}}^{x})(\phi_{AB})\}.
\end{align}
Then the following inequality holds
$
S(B)_{\phi}\geq S(B^{\prime}|X)_{\tau},
$
for $\tau_{XA^{\prime}B^{\prime}}\equiv\sum_{x}p(x)|x\rangle\langle
x|_{X}\otimes\varphi_{A^{\prime}B^{\prime}}^{x}$.
\end{lemma}

The above lemma leads to the following one, which is the statement that the
quantity in \eqref{eq:1W-LOCC-monotone} is monotone with respect to 1W-LOCC\ channels:

\begin{lemma}
\label{lem:1W-LOCC-monotone}Let $\tau_{WFAB}$ be a classical--quantum state,
with classical systems $WF$ and quantum systems $AB$ pure when conditioned on $WF$, and let $\mathcal{M}%
_{AB\rightarrow A^{\prime}B^{\prime}X}$ be a 1W-LOCC channel of the form in
\eqref{eq:1W-LOCC-ch}. Then the following  holds%
\begin{equation}
I(W;BF)_{\tau}+S(B|WF)_{\tau}\geq I(W;B^{\prime}FX)_{\theta}+S(B^{\prime
}|WFX)_{\theta},\notag
\end{equation}
where $\theta_{WFA^{\prime}B^{\prime}X}\equiv\mathcal{M}_{AB\rightarrow
A^{\prime}B^{\prime}X}(\tau_{WFAB})$.
\end{lemma}

\begin{proof}
The inequality $I(W;BF)_{\tau}\geq I(W;B^{\prime}FX)_{\theta}$ holds from data
processing. In more detail, consider that $\theta_{WFB^{\prime}X}$ is equal to
\begin{align}
& =\operatorname{Tr}_{A^{\prime}}\left\{  \sum_{x}(\mathcal{U}_{A\rightarrow
A^{\prime}}^{x}\otimes\mathcal{V}_{B\rightarrow B^{\prime}}^{x})(\tau
_{WFAB})\otimes|x\rangle\langle x|_{X}\right\}  \notag \\
& =\sum_{x}\mathcal{V}_{B\rightarrow B^{\prime}}^{x}(\tau_{WFB})\otimes
|x\rangle\langle x|_{X},
\end{align}
where the last equality follows because each map $\mathcal{U}_{A\rightarrow
A^{\prime}}^{x}$ is trace preserving. So the state $\theta_{WFB^{\prime}X}$
can be understood as arising from the action of the quantum instrument
$\sum_{x}\mathcal{V}_{B\rightarrow B^{\prime}}^{x}\otimes|x\rangle\langle
x|_{X}$ on the state $\tau_{WFB}$, and since this is a channel from $B$ to
$B^{\prime}X$, the data processing inequality applies so that $I(W;BF)_{\tau
}\geq I(W;B^{\prime}FX)_{\theta}$.
The inequality $S(B|WF)_{\tau}\geq S(B^{\prime}|WFX)_{\theta}$ follows from an
application of Lemma~\ref{lem:LOCC-monotone}, by conditioning on the classical
systems $WF$.
\end{proof}

The following lemma places an entropic upper bound on the amount by which the
quantity in \eqref{eq:1W-LOCC-monotone}\ can increase by the action of a
channel $\mathcal{N}_{A\rightarrow B}$:

\begin{lemma}
\label{lem:amortize-bnd}Let $\tau_{WFAB^{\prime}}$ be a classical--quantum
state of the following form:%
\begin{equation}
\tau_{WFAB^{\prime}}=\sum_{w,f}p(w,f)|w\rangle\langle w|_{W}\otimes
|f\rangle\langle f|_{F}\otimes\tau_{AB^{\prime}}^{w,f}.
\end{equation}
Then%
\begin{multline}
I(W;BB^{\prime}F)_{\omega}+S(BB^{\prime}|WF)_{\omega}\\
-\left[  I(W;B^{\prime}F)_{\tau}+S(B^{\prime}|WF)_{\tau}\right]  \leq
S(B)_{\omega},
\end{multline}
where $\omega_{WFBB^{\prime}}\equiv\mathcal{N}_{A\rightarrow B}(\tau
_{WFAB^{\prime}})$.
\end{lemma}

\begin{proof}
Consider that%
\begin{align}
& I(W;BB^{\prime}F)_{\omega}+S(BB^{\prime}|WF)_{\omega}\notag\\
& \qquad -\left[  I(W;B^{\prime
}F)_{\tau}+S(B^{\prime}|WF)_{\tau}\right]  \nonumber\\
& =I(W;BB^{\prime}F)_{\omega}+S(BB^{\prime}|WF)_{\omega} \notag\\
& \qquad -\left[
I(W;B^{\prime}F)_{\omega}+S(B^{\prime}|WF)_{\omega}\right]  \\
& =I(W;B|B^{\prime}F)_{\omega}+S(B|B^{\prime}WF)_{\omega}\\
& =S(B|B^{\prime}F)_{\omega}-S(B|B^{\prime}WF)_{\omega}+S(B|B^{\prime
}WF)_{\omega}\\
& =S(B|B^{\prime}F)_{\omega} \leq S(B)_{\omega}.
\end{align}
All inequalities follow from definitions and applying chain rules for
mutual information and entropy. The final inequality follows because conditioning
does not increase entropy.
\end{proof}

\section{Maximum output entropy bound}

\label{sec:max-out-ent-bnd}

Now that we have identified a quantity that does not increase under
1W-LOCC\ from Bob to Alice and cannot increase by more than the output entropy
of a channel under its action, we can use these properties to establish the
following upper bound on the rate of a feedback-assisted communication protocol:

\begin{theorem}
\label{thm:upper-bnd}For an $(n,M,H,E,\varepsilon)$ protocol for classical
communication over a quantum channel $\mathcal{N}_{A\rightarrow B}$ assisted
by classical feedback, of the form described in
Section~\ref{sec:feedback-protocol}, the following bound applies%
\begin{equation}
\left(  1-\varepsilon\right)  \log_{2}M\leq n\cdot\sup_{\rho:\operatorname{Tr}%
\{H\rho\}\leq E}S(\mathcal{N}(\rho))+h_{2}(\varepsilon
).\label{eq:weak-conv-bnd-feedback}%
\end{equation}

\end{theorem}

\begin{proof}
Let us consider the purified simulation of a given $(n,M,H,E,\varepsilon)$
protocol, as given in Section~\ref{sec:purified-protocol}. We start with
\begin{align}
\log_{2}M  & =I(W;\hat{W})_{\overline{\Phi}}\label{eq:bnd-1}\\
& \leq I(W;\hat{W})_{\rho}+\varepsilon\log_{2}M+h_{2}(\varepsilon
),\label{eq:bnd-2}%
\end{align}
where we have applied the condition in~\eqref{eq:good-protocol-condition}\ and
standard entropy inequalities. Continuing, we find that%
\begin{align}
& I(W;\hat{W})_{\rho}  \notag \\
& \leq I(W;B_{n}\hat{B}_{n}[F_{0}^{n-1}]^{\prime}%
)_{\rho^{(n)}}+S(B_{n}\hat{B}_{n}|[F_{0}^{n-1}]^{\prime}W)_{\rho^{(n)}%
}\label{eq:bnd-3}\\
& =I(W;B_{n}\hat{B}_{n}[F_{0}^{n-1}]^{\prime})_{\rho^{(n)}}+S(B_{n}\hat{B}%
_{n}|[F_{0}^{n-1}]^{\prime}W)_{\rho^{(n)}}\nonumber\\
& \qquad-\left[  I(W;\hat{B}_{1}F_{0}^{\prime})_{\omega^{(1)}}+S(\hat{B}%
_{1}|F_{0}^{\prime}W)_{\omega^{(1)}}\right]  \\
& =I(W;B_{n}\hat{B}_{n}[F_{0}^{n-1}]^{\prime})_{\rho^{(n)}}+S(B_{n}\hat{B}%
_{n}|[F_{0}^{n-1}]^{\prime}W)_{\rho^{(n)}}\nonumber\\
& \qquad-\left[  I(W;\hat{B}_{1}F_{0}^{\prime})_{\omega^{(1)}}+S(\hat{B}%
_{1}|F_{0}^{\prime}W)_{\omega^{(1)}}\right]  \nonumber\\
& \qquad+\sum_{i=2}^{n}I(W;\hat{B}_{i}[F_{0}^{i-1}]^{\prime})_{\omega^{(i)}%
}+S(\hat{B}_{i}|[F_{0}^{i-1}]^{\prime}W)_{\omega^{(i)}}\nonumber\\
& \qquad\qquad-\left[  I(W;\hat{B}_{i}[F_{0}^{i-1}]^{\prime})_{\omega^{(i)}%
}+S(\hat{B}_{i}|[F_{0}^{i-1}]^{\prime}W)_{\omega^{(i)}}\right] \notag .
\end{align}
The first inequality follows from data processing and non-negativity of
entropy. The first equality follows because $I(W;\hat{B}_{1}F_{0}^{\prime
})_{\omega^{(1)}}+S(\hat{B}_{1}|F_{0}^{\prime}W)_{\omega^{(1)}}=0$ for the
initial state $\omega_{W\hat{A}_{1}A_{1}\hat{B}_{1}F_{0}^{\prime}}^{(1)}$
(there is no classical correlation between $W$ and $\hat{B}_{1}F_{0}^{\prime}$,
and the state on system $\hat{B}_{1}$ is pure when conditioned on
$F_{0}^{\prime}W$). The last equality follows by adding and subtracting the
same term. Continuing, we find that the quantity in the last line above is
bounded as%
\begin{align}
& \leq I(W;B_{n}\hat{B}_{n}[F_{0}^{n-1}]^{\prime})_{\rho^{(n)}}+S(B_{n}\hat
{B}_{n}|[F_{0}^{n-1}]^{\prime}W)_{\rho^{(n)}}\nonumber\\
& \qquad-\left[  I(W;\hat{B}_{1}F_{0}^{\prime})_{\omega^{(1)}}+S(\hat{B}%
_{1}|F_{0}^{\prime}W)_{\omega^{(1)}}\right]  \nonumber\\
& \qquad+\sum_{i=2}^{n}I(W;B_{i-1}\hat{B}_{i-1}[F_{0}^{i-2}]^{\prime}%
)_{\rho^{(i-1)}}\notag\\
& \qquad\qquad +S(B_{i-1}\hat{B}_{i-1}|[F_{0}^{i-2}]^{\prime}W)_{\rho
^{(i-1)}}\nonumber\\
& \qquad\qquad-\left[  I(W;\hat{B}_{i}[F_{0}^{i-1}]^{\prime})_{\omega^{(i)}%
}+S(\hat{B}_{i}|[F_{0}^{i-1}]^{\prime}W)_{\omega^{(i)}}\right] \notag \\
& =\sum_{i=1}^{n}I(W;B_{i}\hat{B}_{i}[F_{0}^{i-1}]^{\prime})_{\rho^{(i)}%
}+S(B_{i}\hat{B}_{i}|[F_{0}^{i-1}]^{\prime}W)_{\rho^{(i)}}\nonumber\\
& \qquad-\left[  I(W;\hat{B}_{i}[F_{0}^{i-1}]^{\prime})_{\omega^{(i)}}%
+S(\hat{B}_{i}|[F_{0}^{i-1}]^{\prime}W)_{\omega^{(i)}}\right]  
\label{eq:bnd-to-connect-app}
\\
& \leq\sum_{i=1}^{n}S(B_{i})_{\rho^{(i)}}
\leq n S(B)_{\mathcal{N}(\overline{\omega})}
\leq n\sup_{\rho:\operatorname{Tr}\{H\rho\}\leq E}S(\mathcal{N}%
(\rho)).\label{eq:bnd-4}%
\end{align}
The first inequality follows from Lemma~\ref{lem:1W-LOCC-monotone}. The first
equality follows by collecting terms. The second inequality follows from
Lemma~\ref{lem:amortize-bnd}. The third inequality follows from concavity of
entropy and the definition of  $\overline{\omega}_{A}$ in
\eqref{eq:energy-constraint}. The final inequality follows from the energy
constraint in \eqref{eq:energy-constraint}, and by optimizing over all input
states that satisfy this energy constraint. By combining
\eqref{eq:bnd-1}--\eqref{eq:bnd-2}\ and \eqref{eq:bnd-3}--\eqref{eq:bnd-4}, we
arrive at \eqref{eq:weak-conv-bnd-feedback}.
\end{proof}

In Appendix~\ref{app:convex-comb}, we show how to extend this result to the maximum average output entropy:
\begin{theorem}
\label{cor:upper-bnd}
Let $\mathcal{N}_{A\rightarrow B} = \sum_x p_X(x) \mathcal{N}^x_{A\rightarrow B}$, where
$p_X$ is a probability distribution and $\{ \mathcal{N}^x_{A\rightarrow B}\}_x$ is a set of channels. For an $(n,M,H,E,\varepsilon)$ protocol for classical
communication over the channel $\mathcal{N}_{A\rightarrow B}$ assisted
by classical feedback, of the form described in
Section~\ref{sec:feedback-protocol}, the following bound applies%
\begin{equation}
\left(  1-\varepsilon\right)  \log_{2}M\leq n\cdot\sup_{\rho:\operatorname{Tr}%
\{H\rho\}\leq E}\sum_x p_X(x) S(\mathcal{N}^x(\rho))+h_{2}(\varepsilon
).\notag
\end{equation}
\end{theorem}

\section{Examples}

\label{sec:examples}

From the upper bound in Theorem~\ref{thm:upper-bnd}, we conclude that the
feedback-assisted capacity of a noiseless qudit channel of dimension $d$ is
log$_{2}d$. 

Furthermore, consider a pure-loss bosonic channel \cite{GGLMSY04} with transmissivity $\eta\in\left[  0,1\right]  $.
Taking the Hamiltonian as the photon number operator and energy constraint $N_{S}\geq0$, it is known from \cite{GGLMSY04}
that this channel's energy constrained classical capacity and maximum output entropy are equal to 
$g(\eta N_{S})$, where $g(x)\equiv\left(  x+1\right)
\log_{2}(x+1)-x\log_{2}x$. Applying these results and Theorem~\ref{thm:upper-bnd}, we conclude that classical feedback does not increase 
the energy-constrained classical capacity of the pure-loss bosonic channel.

A quantum erasure channel is defined as $\rho \to (1-p) \rho + p |e \rangle \langle e|$ for $p\in[0,1]$, where $\rho$ is the state of a $d$-dimensional input system and $|e \rangle \langle e|$ is an erasure state orthogonal to all inputs. Applying Theorem~\ref{cor:upper-bnd}, we conclude that the classical capacity of the erasure channel assisted by classical feedback is equal to $(1-p) \log_2 d$, so that classical feedback does not increase the classical capacity of the erasure channel.

\section{Conclusion}

\label{sec:conclusion}

Our main result is that the maximum average output entropy of a quantum channel is
an upper bound on its classical capacity assisted by classical feedback. Note that the bound is a weak converse bound. 
Going forward from here, it would be good to find strong converse and tighter bounds on the classical capacity assisted by classical feedback.


{\footnotesize \textbf{Acknowledgements.} We acknowledge  discussions with Xin Wang, Patrick Hayden, and Tsachy Weissman.
DD is supported by a National
Defense Science and Engineering Graduate Fellowship.
YQ is supported by a Stanford Graduate Fellowship and a National University of
Singapore Overseas Graduate Scholarship. 
PWS is supported by the NSF under Grant No. CCF-1525130 and through the NSF Science and Technology Center for Science of Information under Grant No.~CCF0-939370.
MMW acknowledges  NSF grant no.~1350397. DD and MMW  thank God for all His provisions.}

\bibliographystyle{IEEEtran}
\bibliography{Ref}

\begin{thebibliography}{10}
\providecommand{\url}[1]{#1}
\csname url@samestyle\endcsname
\providecommand{\newblock}{\relax}
\providecommand{\bibinfo}[2]{#2}
\providecommand{\BIBentrySTDinterwordspacing}{\spaceskip=0pt\relax}
\providecommand{\BIBentryALTinterwordstretchfactor}{4}
\providecommand{\BIBentryALTinterwordspacing}{\spaceskip=\fontdimen2\font plus
\BIBentryALTinterwordstretchfactor\fontdimen3\font minus
  \fontdimen4\font\relax}
\providecommand{\BIBforeignlanguage}[2]{{%
\expandafter\ifx\csname l@#1\endcsname\relax
\typeout{** WARNING: IEEEtran.bst: No hyphenation pattern has been}%
\typeout{** loaded for the language `#1'. Using the pattern for}%
\typeout{** the default language instead.}%
\else
\language=\csname l@#1\endcsname
\fi
#2}}
\providecommand{\BIBdecl}{\relax}
\BIBdecl

\bibitem{S56}
C.~Shannon, ``The zero error capacity of a noisy channel,'' \emph{IRE
  Transactions on Information Theory}, vol.~2, no.~3, pp. 8--19, September
  1956.

\bibitem{B04}
G.~Bowen, ``Quantum feedback channels,'' \emph{IEEE Transactions on Information
  Theory}, vol.~50, no.~10, pp. 2429--2434, October 2004.

\bibitem{PhysRevLett.83.3081}
C.~H. Bennett, P.~W. Shor, J.~A. Smolin, and A.~V. Thapliyal,
  ``Entanglement-assisted classical capacity of noisy quantum channels,''
  \emph{Physical Review Letters}, vol.~83, no.~15, pp. 3081--3084, October
  1999.

\bibitem{ieee2002bennett}
------, ``Entanglement-assisted capacity of a quantum channel and the reverse
  {Shannon} theorem,'' \emph{IEEE Transactions on Information Theory}, vol.~48,
  no.~10, pp. 2637--2655, October 2002.

\bibitem{Hol01a}
A.~S. Holevo, ``On entanglement assisted classical capacity,'' \emph{Journal of
  Mathematical Physics}, vol.~43, no.~9, pp. 4326--4333, September 2002.

\bibitem{BDHSW09}
C.~H. Bennett, I.~Devetak, A.~W. Harrow, P.~W. Shor, and A.~Winter, ``The
  quantum reverse {Shannon} theorem and resource tradeoffs for simulating
  quantum channels,'' \emph{IEEE Transactions on Information Theory}, vol.~60,
  no.~5, pp. 2926--2959, May 2014.

\bibitem{CMW14}
T.~Cooney, M.~Mosonyi, and M.~M. Wilde, ``Strong converse exponents for a
  quantum channel discrimination problem and quantum-feedback-assisted
  communication,'' \emph{Communications in Mathematical Physics}, vol. 344,
  no.~3, pp. 797--829, June 2016.

\bibitem{BN05}
G.~Bowen and R.~Nagarajan, ``On feedback and the classical capacity of a noisy
  quantum channel,'' \emph{IEEE Transactions on Information Theory}, vol.~51,
  no.~1, pp. 320--324, January 2005.

\bibitem{DingW18}
D.~Ding and M.~M. Wilde, ``Strong converse exponents for the feedback-assisted
  classical capacity of entanglement-breaking channels,'' \emph{Problems of
  Information Transmission}, vol.~54, no.~1, pp. 1--19, January 2018.

\bibitem{BDSS06}
C.~H. Bennett, I.~Devetak, P.~W. Shor, and J.~A. Smolin, ``Inequalities and
  separations among assisted capacities of quantum channels,'' \emph{Physical
  Review Letters}, vol.~96, no.~15, p. 150502, April 2006.

\bibitem{SS09}
G.~Smith and J.~A. Smolin, ``Extensive nonadditivity of privacy,''
  \emph{Physical Review Letters}, vol. 103, no.~12, p. 120503, September 2009.

\bibitem{PhysRevA.55.1613}
S.~Lloyd, ``Capacity of the noisy quantum channel,'' \emph{Physical Review A},
  vol.~55, no.~3, pp. 1613--1622, March 1997.

\bibitem{capacity2002shor}
P.~W. Shor, ``The quantum channel capacity and coherent information,'' in
  \emph{Lecture Notes, MSRI Workshop on Quantum Computation}, 2002.

\bibitem{ieee2005dev}
I.~Devetak, ``The private classical capacity and quantum capacity of a quantum
  channel,'' \emph{IEEE Transactions on Information Theory}, vol.~51, no.~1,
  pp. 44--55, January 2005, arXiv:quant-ph/0304127.

\bibitem{TGW14}
M.~Takeoka, S.~Guha, and M.~M. Wilde, ``The squashed entanglement of a quantum
  channel,'' \emph{IEEE Transactions on Information Theory}, vol.~60, no.~8,
  pp. 4987--4998, August 2014, arXiv:1310.0129.

\bibitem{CM17}
M.~Christandl and A.~M\"{u}ller-Hermes, ``Relative entropy bounds on quantum,
  private and repeater capacities,'' \emph{Communications in Mathematical
  Physics}, vol. 353, no.~2, pp. 821--852, July 2017.

\bibitem{KW17}
E.~Kaur and M.~M. Wilde, ``Amortized entanglement of a quantum channel and
  approximately teleportation-simulable channels,'' \emph{Journal of Physics
  A}, vol.~51, no.~3, p. 035303, January 2018.

\bibitem{BHKW18}
M.~Berta, C.~Hirche, E.~Kaur, and M.~M. Wilde, ``Amortized channel divergence
  for asymptotic quantum channel discrimination,'' August 2018,
  arXiv:1808.01498.

\bibitem{BDSW96}
C.~H. Bennett, D.~P. DiVincenzo, J.~A. Smolin, and W.~K. Wootters,
  ``Mixed-state entanglement and quantum error correction,'' \emph{Physical
  Review A}, vol.~54, no.~5, pp. 3824--3851, November 1996.

\bibitem{GGLMSY04}
V.~Giovannetti, S.~Guha, S.~Lloyd, L.~Maccone, J.~H. Shapiro, and H.~P. Yuen,
  ``Classical capacity of the lossy bosonic channel: The exact solution,''
  \emph{Phys. Rev. Lett.}, vol.~92, no.~2, p. 027902, January 2004.

\end{thebibliography}

\pagebreak

\appendices

\section{Maximum average output entropy bound for probabilistic mixture of channels}

\label{app:convex-comb}

In this appendix, we provide a simple proof of Theorem~\ref{cor:upper-bnd}. The main idea behind the proof is to observe that any feedback-assisted protocol of the form discussed in Section~\ref{sec:feedback-protocol}, which is for communication over a probabilistic mixture channel $\mathcal{N}_{A\rightarrow B} = \sum_z p_Z(z) \mathcal{N}^z_{A\rightarrow B}$, has a simulation of the following form:
\begin{enumerate}

\item Before the $i$th use of the channel $\mathcal{N}_{A\rightarrow B}$ in the feedback-assisted protocol, Bob selects a random variable $Z_i$ independently according to the distribution $p_Z$. He transmits $Z_i$ over the classical feedback channel to Alice.

\item Each channel use $\mathcal{N}_{A\rightarrow B}$ from the original protocol is replaced by a simulation in terms of another channel $\mathcal{M}_{AZ'\to B}$, which accepts a quantum input on system $A$ and a classical input on system $Z'$. Conditioned on the value $z$ in system $Z'$, the channel $\mathcal{M}_{AZ'\to B}$ applies $\mathcal{N}^z_{A\rightarrow B}$ to the quantum system $A$. Thus, if the random variable $Z \sim p_Z$ is fed into the input system $Z'$ of $\mathcal{M}_{AZ'\to B}$, then the channel $\mathcal{M}_{AZ'\to B}$ is indistinguishable from the original channel $\mathcal{N}_{A\rightarrow B}$.

\item Alice feeds a copy of the classical random variable $Z_{i}$ into the $i$th use of the channel $\mathcal{M}_{AZ'\to B}$.

\item All other aspects of the protocol are executed in the same way as before. Namely, even though it would be an advantage to Alice to modify her encodings and Bob to modify later decodings based on the realizations of $Z_i$, they do not do so, and they instead blindly operate all other aspects of the simulation protocol as they are in the original protocol.

\end{enumerate}
Our goal now is to establish the inequality in Theorem~\ref{cor:upper-bnd},  relating the $n$, $M$, $E$, $\varepsilon$ parameters of the original $(n,M,H,E,\varepsilon)$ protocol by using the above simulation.

The main observation to make from here is that the same proof from Lemma~\ref{lem:amortize-bnd} gives the following bound:
\begin{multline}
I(W;BB^{\prime}FZ)_{\omega}+S(BB^{\prime}|WFZ)_{\omega}\\
-\left[  I(W;B^{\prime}FZ)_{\tau}+S(B^{\prime}|WFZ)_{\tau}\right]  \leq
S(B|Z)_{\omega},
\label{eq:new-ineq-amortize}
\end{multline}
where
$\omega_{WFZBB^{\prime}}$ is the following state:
\begin{equation}
\omega_{WFZBB^{\prime}}  \equiv\mathcal{M}_{AZ'\rightarrow B}(\tau
_{WFZZ'AB^{\prime}}) 
\end{equation}
\vspace{-.3in}
\begin{multline}
\tau_{WFZZ'AB^{\prime}}  \equiv \\
\sum_{w,f,z}p(w,f,z)|w,f,z,z\rangle\langle w,f,z,z|_{W,F,Z,Z'}\otimes\tau_{AB^{\prime}}^{w,f,z}.
\end{multline}
This follows by grouping $Z$ with $F$, but then discarding only $F$ and $B'$ at the end of the proof.
We then apply this bound, and the same reasoning in the proof of Theorem~\ref{thm:upper-bnd}, except that the variables $Z_0, \ldots, Z_i$ are grouped together with the feedback variables $[F_0^{i-1}]'$ and then the same reasoning in \eqref{eq:bnd-3}--\eqref{eq:bnd-to-connect-app} applies. At this point, we invoke \eqref{eq:new-ineq-amortize} and find that
\begin{equation}
(1-\varepsilon) \log_2 M \leq \sum_{i=1}^n S(B_i|Z_{i})_{\rho^{(i)}} + h_2(\varepsilon).
\end{equation}

We can then bound the sum over entropies as follows:
\begin{align}
 \sum_{i=1}^n S(B_i|Z_{i})_{\rho^{(i)}} & \leq n S(B|Z)_{\overline{\rho}} \\
& = n \sum_z p_Z (z) S( \mathcal{N}^z (\overline{\omega}))\\
& \leq n \sup_{\rho:\operatorname{Tr}%
\{H\rho\}\leq E}\sum_z p_Z(z) S(\mathcal{N}^z(\rho)).
\end{align}

The first inequality is by concavity of conditional entropy, and the conditional entropy is defined on the averaged channel output state over uses $\overline{\rho}_{BZ} \equiv \sum_z p_Z (z) |z\rangle \langle z| \otimes \mathcal{N}^z ( \overline{\omega})$, $\overline{\omega}_A = \frac{1}{n}\sum_{i=1}^n \omega^{(i)}_{A_i}$. The second equality is by definition of conditional entropy. The third inequality follows from optimizing over states that satisfy the energy constraint in \eqref{eq:energy-constraint}. This concludes the proof of Theorem~\ref{cor:upper-bnd}. 

\end{document}